\documentclass[sn-mathphys,Numbered]{sn-jnl}% Math and Physical Sciences 
\usepackage{graphicx}%
\usepackage{multirow}%
\usepackage{amsmath,amssymb,amsfonts}%
\usepackage{amsthm}%
\usepackage{mathrsfs}%
\usepackage[title]{appendix}%
\usepackage{xcolor}%
\usepackage{textcomp}%
\usepackage{manyfoot}%
\usepackage{booktabs}%
\usepackage{algorithm}%
\usepackage{algorithmicx}%
\usepackage{algpseudocode}%
\usepackage{listings}%

\usepackage{subcaption} % Enable subfigure-environment
\usepackage{tikz}
\usetikzlibrary{spy}
\usepackage{amsmath}
\usepackage{amssymb}
\usepackage{bm}
\usepackage{natbib}

\newcommand{\R}{\mathbb{R}}
\newcommand{\grad}{\bm{\nabla}}
\newcommand{\Div}{\mbox{\rm{\bf{div\,}}}}

\newcommand{\ut}[1]{\bm u_{#1}}
\newcommand{\eps}{\epsilon}

\newcommand{\intd}{\, \mathrm{d} }
\newcommand{\conste}{\, \mathrm{e} }

\newtheorem{theorem}{Theorem}%  meant for continuous numbers
\newtheorem{proposition}[theorem]{Proposition}% 
\newtheorem{corollary}[theorem]{Corollary}%
\begin{document}

\title[Generalised Diffusion Probabilistic Scale-Spaces]{Generalised 
    Diffusion Probabilistic Scale-Spaces}

%% \author*[1,2]{\pfx{Dr} \fnm{Joergen W.} \spfx{van der} \sur{Ploeg} \sfx{IV} 
%%\tanm{Poet Laureate} 

\author*[1]{\fnm{Pascal} \sur{Peter}}\email{peter@mia.uni-saarland.de}


\affil*[1]{\orgdiv{Mathematical Image Analysis Group, Faculty of 
        Mathematics and Computer Science}, \orgname{Saarland University}, 
    \orgaddress{\street{Campus E1.7}, \city{Saarbr\"ucken}, 
        \postcode{66041}, \state{Saarland}, \country{Germany}}}

\abstract{
    Diffusion probabilistic models excel at sampling new images from 
    learned distributions. Originally motivated by drift-diffusion concepts 
    from physics, they apply image perturbations such as noise and blur in a 
    forward process that results in a tractable probability distribution. A 
    corresponding learned reverse process generates images and can be 
    conditioned on side information, which leads to a wide variety of practical 
    applications.
    
    Most of the research focus currently lies on practice-oriented 
    extensions. In contrast, the theoretical background remains largely 
    unexplored, in particular the relations to drift-diffusion. In order to 
    shed light on these connections to classical image filtering, we propose a 
    generalised scale-space theory for diffusion probabilistic models. 
    Moreover, we show conceptual and empirical connections to diffusion and 
    osmosis filters.}

\keywords{diffusion probabilistic models, scale-spaces, 
    drift-diffusion,       
    osmosis}



\maketitle

\section{Introduction}

Diffusion probabilistic models~\cite{SMDH13} have recently risen to
the state-of-the-art in image generation, surpassing  generative adversarial 
networks~\cite{GPMX14} in popularity. In addition to significant research 
activity, the availability of pre-trained latent diffusion 
networks~\cite{RBLE+22} has also brought diffusion models to widespread 
public attention~\cite{DN21}. Practical applications are numerous, including 
the generation of convincing, high fidelity images from text prompts or partial 
image data.

Initial diffusion probabilistic models~\cite{SMDH13,HJA20,KSBH21,SE19, 
    SE20, SDME21, SSKK+21} relied on a forward drift-diffusion process that 
gradually perturbs input images with noise and can be reversed by deep 
learning. 
Recently, it has been shown that the concrete mechanism that 
gradually destroys information in the forward process has a significant impact 
on the image generation by the reverse process. Alternative proposed image 
degradations include blur \cite{BBCL+22}, combinations of noise and blur 
\cite{DDTDM23, RHS23, HS22}, or image masking~\cite{DDTDM23}.

So far, diffusion probabilistic research was mostly of practical 
nature. Some theoretical contributions established connections to other
fields such as score-matching \cite{SE19, SE20, SDME21, SSKK+21}, variational 
autoencoders~\cite{KSBH21}, or normalising flows~\cite{SHH23}. Diffusion 
    probabilistic models have been initially motivated~\cite{SMDH13} by 
drift-diffusion, a well-known process in physics. However, its connections to 
other physics-inspired methods remain mostly unexplored. Closely related  
concepts have a long 
tradition in model-based visual computing, such as osmosis filtering proposed 
by~\citet{WHBV13}. In addition, there is a wide variety of diffusion-based 
scale-spaces~\cite{AGLM93,Ii62,SchW98,We97}. Conceptually, these scale-spaces 
embed given images into a family of simplified versions. This resembles the 
gradual removal of image features in the forward process of diffusion 
    probabilistic models. 

Despite this multitude of connections, there is a distinct lack of systematic 
analysis of diffusion probabilistic models from a scale-space 
perspective. 
This is particularly surprising due to the impact 
of the forward process on the generative performance~\cite{HS22,RHS23}. It 
indicates that a  deeper understanding of the information reduction could 
also lead to further practical improvements in the future.

\subsection{Our Contribution}

With our previous conference publication~\cite{Pe23} we made first steps to 
bridge this gap between the scale-space and deep learning communities. To this 
end, we introduced first generalised scale-space concepts for diffusion 
    probabilistic models. In this work, we further explore the theoretical 
background of this successful paradigm in deep learning. 
In contrast to traditional scale-spaces, 
we consider the evolution of probability distributions instead of images. 
Despite this departure from conventional families of images, we can show  
scale-space properties in  the sense of \citet{AGLM93}. These include 
architectural properties, invariances, and entropy-based measures of 
simplification.

In addition to our previous findings~\cite{Pe23}, our novel contributions 
include
\begin{itemize}
    \item a generalisation of our scale-space theory for diffusion 
    probabilistic models (DPMs) which includes both variance-preserving 
    and variance-exploding approaches,
    \item generalised scale-space properties for the reverse process of DPM,
    \item a scale-space theory for inverse heat dissipation~\cite{RHS23} and 
    blurring diffusion~\cite{HS22},
    \item and a significantly extended theoretical and empirical comparison of 
    three diffusion probabilistic models to homogeneous 
    diffusion~\cite{Ii62} and osmosis filtering~\cite{WHBV13}.
\end{itemize}

\section{Related Work}

Besides diffusion probabilistic models themselves, two additional 
research areas are relevant for our own work. 
Since we adopt a scale-space perspective, 
classical scale-space research acts as the foundation for our generalised 
theory.
Furthermore, we discuss connections to osmosis filters, which have a tradition 
in model-based visual computing.

\subsection{Diffusion Probabilistic Models}

Large parts of our scale-space theory are based on the work of 
\citet{SMDH13}, which pioneered diffusion probabilistic models (DPMs). The 
latent diffusion model by \citet{RBLE+22} was integral for the gain of 
popularity of this approach. The public availability of code and trained models 
sparked many practical applications and caused a shift~\cite{DN21} away from 
generative adversarial networks~\cite{GPMX14}. Applications range from image 
generation~\cite{SMDH13,HJA20,DN21,RBLE+22} over image 
inpainting~\cite{LDRY+22,SCCL+22}, 
super-resolution~\cite{HSCF+22}, segmentation~\cite{WSBV+22}, and 
deblurring~\cite{RDTGM+22} to the generation of different types of media, 
including video~\cite{HSCF+22} and audio~\cite{KPHZ+21}.

Following the same principles as early DPMs, most approaches rely on adding 
noise in the forward process. Since coarse-to-fine strategies were shown to 
improve DPMs~\cite{ND21,HSCF+22}, they inspired \citet{LCKY22} to use blurring 
with Gaussian convolution instead. \citet{RHS23} leveraged the 
equivalence of Gaussian blur to homogeneous diffusion~\cite{Ii62} to establish 
inverse heat dissipation models. 
They only add small amounts of observation noise, while the 
model of~\citet{HS22} generalises this process and allows more substantial 
contributions of noise. Due to the close relations to both classical diffusion 
and osmosis filtering, we address both of these models in detail in 
Section~\ref{sec:blurringdiffusion}. Image masking was recently also proposed 
as an alternative degradation in the forward process by \citet{DDTDM23}.

We also rely on a more general version of the original DPM model that was 
proposed by \citet{KSBH21}. In addition, they introduced the notion of 
variational diffusion models and showed connections to variational 
autoencoders. 
Interestingly, early results of \citet{Vi11} established connections between 
denoising autoencoders and score-matching. Later publications 
showed more direct relations of score-based approaches to DPM~\cite{HJA20}. 
\citet{SE19} related diffusion models to score-based 
Langevin dynamics with later follow-up results \cite{SE20, SDME21, SSKK+21} 
that also include continuous time processes.

Theoretical contributions that are related to our own work are rare. Recently, 
\citet{SHH23} connected diffusion probabilistic models to a multitude of 
different concepts under the common framework of normalising flows. 
This includes relations to osmosis filtering, but not from a scale-space 
perspective. 
\citet{FRRH+23} have investigated functional diffusion processes as a 
time-continuous generalisation of classical DPMs. They also allow both noise 
and blur as image degradations. In contrast, our models are time- and 
space-discrete.

\subsection{Scale-Spaces}

The second major field of research that forms the foundation of our 
contribution is scale-space theory. Scale-spaces have a long tradition in 
visual computing. Most of them rely on partial 
differential~\cite{AGLM93,Ii62,Li11,SchW98,We97} or pseudo-differential 
equations~\cite{DFGH04,SW16}, but they have also been 
considered for wavelets~\cite{CL01}, sparse inpainting-based image 
representations~\cite{CPW19}, or hierarchical quantisation 
operators~\cite{Pe21}. Such classical 
scale-spaces describe the evolution of an input image over multiple scales, 
which gradually simplifies the image. Since they obey a hierarchical structure 
and provide guarantees for simplification, they allow to analyse image features 
that are specific to individual scales. This makes them useful for tasks such 
as corner detection~\cite{AM94a}, modern invariant feature 
descriptors~\cite{ABD12,Lo04a}, or motion 
estimation~\cite{AWS99a,DWBZ12,AMJBT20}.

General principles for classical scale-spaces are vital for our contributions.
They form the foundation for our generalised scale-space theory for DPM. We 
establish architectural, invariance, and information reduction properties in 
the sense of Alvarez et 
al.~\cite{AGLM93} for this new setting. In Section~\ref{sec:probdiffscale} we 
also mention where we drew inspiration from this contribution and other 
sources~\cite{Ii62,We97} in more detail.

There are many different classes of scale-spaces, originating from the early 
work by Iijima~\cite{Ii62}, which was later popularised by Witkin~\cite{Wi83}. 
They proposed a scale-space that can be interpreted as evolutions according to 
homogeneous diffusion. These initial Gaussian scale-spaces 
\cite{AGLM93,Ii62,Wi83,Li94,SNFJ96,Fl13} have been generalised with 
pseudodifferential operators~\cite{DFGH04,SW16,FS01,BDW05} or nonlinear 
diffusion equations~\cite{SchW98,We97}. Moreover, a comprehensive theory for 
shape analysis exists in the form of morphological 
scale-spaces~\cite{AGLM93,BM92,BS94,CS96,KS96,ST93}. Wavelet shrinkage as 
a form of blurring \cite{CL01} and sparse image 
representations~\cite{CPW19,Pe21} have been 
considered from a scale-space perspective as well. Among this wide variety of 
different options, for 
us, the original Gaussian scale-space is still the most relevant. It is 
    closely related to the blurring diffusion processes we consider in 
Section~\ref{sec:blurringdiffusion}.

Our novel class of stochastic scale-spaces considers families of 
probability distributions instead of sequences of images. Conceptually similar 
approaches are rare. The Ph.D. thesis of \citet{Ma00} proposes a 
stochastic concept, which also considers drift-diffusion. However, it is not 
related to deep learning and simplifies images in a different way. Instead of 
adding noise or blur, it shuffles image pixels. Similarly, 
\citet{KV99} proposed ``locally orderless images'', a local pixel shuffling as 
an alternative to blur. Other probabilistic scale-space concepts are only 
broadly related. There have been theoretical considerations of connections 
between diffusion filters and the statistics of natural images \cite{Pe03} and 
practical applications in stem cell differentiation \cite{HKMR+16}.

In parallel to our conference publication~\cite{Pe23}, \citet{ZPKC23} have used 
homogeneous diffusion scale-spaces on probability densities. However, they 
learn image priors via denoising score matching with practical applications to 
image denoising. We on the other hand focus on the scale-space theory of 
generative diffusion probabilistic models.

\subsection{Osmosis Filtering}

In visual computing, osmosis filtering is a successful class of filters that 
has been introduced by \citet{WHBV13} and generalises diffusion 
filtering~\cite{We97}. Even though it creates deterministic image evolutions, 
it is connected to statistical physics. Namely, it is closely related to the 
Fokker-Planck equation \cite{Ri84} and by extension also to 
Langevin formulations and the Beltrami flow \cite{So01b}. This suggests that 
there could also be connections to diffusion probabilistic models.

Since such connections to drift-diffusion also apply to diffusion 
probabilistic models, we investigate connections between these approaches in 
Section~\ref{sec:osmosis}. There, we also discuss the continuous theory for 
osmosis filters as it was originally proposed by \citet{WHBV13} and later 
extended by \citet{Sc18}. \citet{VHWS13} introduced both the corresponding 
discrete theory and a fast implicit solver which we use for our experiments.

Osmosis filters are well suited to integrate conflicting information from 
multiple images, which makes them an excellent tool for image editing 
\cite{DMM18,VHWS13,WHBV13}. Additionally, they have been successfully used for 
shadow removal \cite{DMM18,PCCSW19,WHBV13}, the fusion of spectral 
images~\cite{PCD18,PCBP+20}, and image blending~\cite{BPW23}. There are also 
applications for osmosis that do 
not deal with images. Notably, \citet{HBWV12} used osmosis to enhance numerical 
methods and considered a Markov chain formulation. While we deal with Markov 
processes in this paper, our interpretation of osmosis and the context in which 
we use it is significantly different.

There are also conceptually similar methods in visual computing that are also 
connected to drift--diffusion and predate osmosis. Namely, \citet{HBVW09} 
proposed a lattice Boltzmann model for dithering. Other broadly related 
filters are the directed diffusion models of \citet{IN93} and the covariant 
derivative approach of \citet{Ge06}.

\section{Organisation of the Paper}

We introduce the basic ideas of diffusion probabilistic models in 
Section~\ref{sec:probdiff}, including Markov formulations for the forward and 
reverse processes. Based on these foundations, we propose generalised 
scale-space properties for three classes of probabilistic forward diffusion 
processes in Section~\ref{sec:probdiffscale} and briefly address reverse 
processes as well. As a link to classical scale-spaces and deterministic image 
filters, we investigate relations of diffusion probabilistic models to 
homogeneous diffusion and osmosis filtering in Section~\ref{sec:discussion}. We 
conclude with a discussion and an outlook in Section~\ref{sec:conclusion}.

%%%%%%%%%%%%%%%%%%%%%%%%%%%%%%%%%%%%%%%%%%%%%%%%%%%%%%%%%%%%%%%%%%%%%%%%%%%%%%%%

\section{Diffusion Probabilistic Models}
\label{sec:probdiff}

Diffusion probabilistic models~\cite{SMDH13} are generative approaches 
    which have the goal to create new samples from a desired distribution. 
This distribution is unknown except for a set of given representatives. 
For image processing purposes, this 
training data typically consists of a set of images $\bm f_1,...,\bm 
f_{n_t} \in \R^{n_x n_y n_c}$ with $n_c$ colour channels of size $n_x \times 
n_y$ and $n=n_x n_y n_c$ pixels. From 
a stochastic point of view, these images are realisations of a random variable 
$\bm F$ with an unknown probability density function $p(\bm F)$. 
DPMs aim to sample from this \emph{target distribution}.

\begin{figure*}[t]
    \centering
    \includegraphics[width=0.8\linewidth]{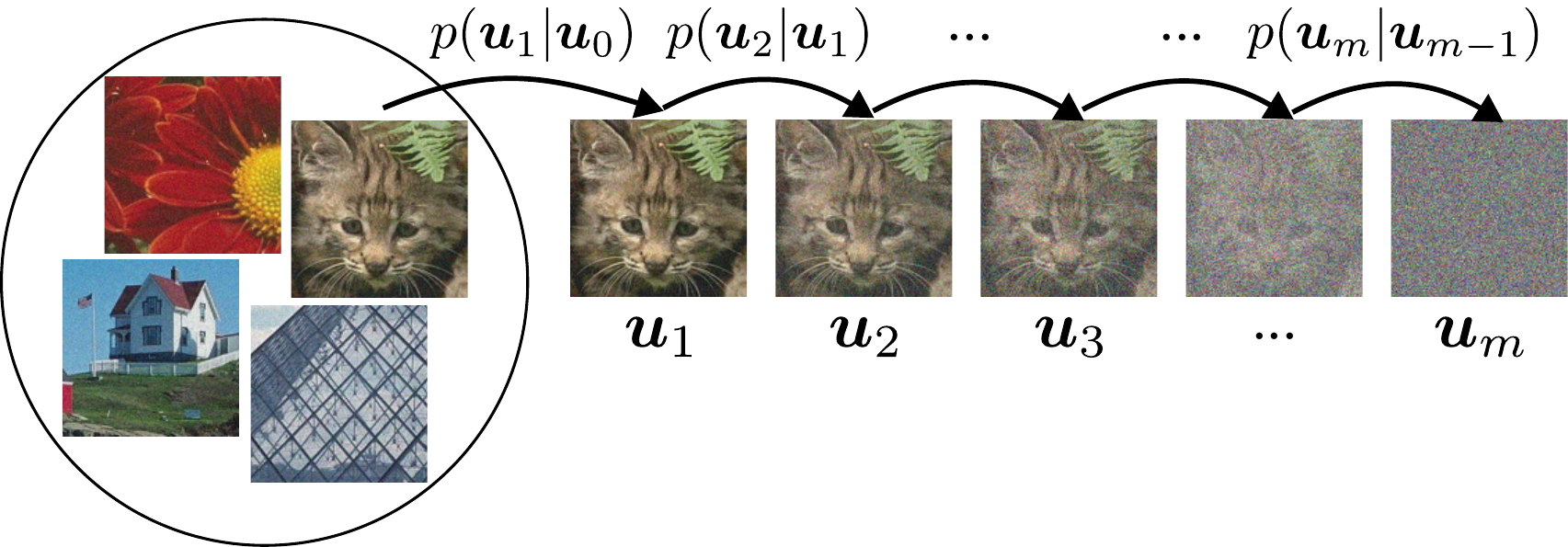}\\[2mm]
    \caption{\textbf{Forward DPM Trajectory.} Starting from each sample 
        of the initial distribution $p(\bm u_0)$, infinitely many trajectories 
        exist. In each step of the trajectory, noise is added according to the 
        transition probability $p(\bm u_i | \bm u_{i-1})$.
        \label{fig:trajectory}}
\end{figure*}

\subsection{The Forward Process}

In a first step, the so-called \emph{forward process}, diffusion 
    probabilistic models map from the target distribution to a simpler 
distribution. 
While there are many alternatives, the standard normal distribution 
$\mathcal{N}(\bm 0, \bm I)$ is a typical choice. 
Here, $\bm I \in \R^{n \times n}$ denotes the 
unit matrix. Thus, the forward process takes training images 
as an input and maps them to samples of multivariate Gaussian noise. These 
noise samples act as seeds for the \emph{reverse process}. Like other 
generative models such as generative adversarial networks~\cite{GPMX14}, it 
maps from samples of this simple distribution back to the approximate target 
distribution. These diffusion probabilistic models (DPMs) create image 
structures from pure noise.

In Section~\ref{sec:blurringdiffusion}, we also present alternatives to purely 
noise-based forward processes. Typically, the forward process is 
straightforward both conceptually and in its implementation. For practical 
tasks, a major challenge is the estimation of the corresponding reverse 
process, which is implemented with deep learning. 
Our focus lies mostly on a theoretical analysis of the forward process from a 
scale-space perspective. 
We also discuss the reverse process in Section~\ref{sec:ddpmbackward}, but due 
to its approximate nature, theoretical results are less comprehensive. The 
design of the forward process also has significant impact on the performance of 
the generative model~\cite{RHS23,HS22}. 
Thus, it also constitutes a more attractive direction for a scale-space focused 
investigation: Understanding the nature of existing forward processes might 
allow to carry over useful properties from existing classical scale-spaces.

Therefore, in Section~\ref{sec:probdiffscale}, we show that a wide variety of 
existing diffusion probabilistic models fulfil generalised 
scale-space properties. 
This new class of scale-spaces differs significantly from classical 
approaches, since it does not consider the evolution of images, but of 
probability distributions instead. First, we need to establish a mathematical 
definition of the probabilistic forward process on a time-dependent random 
variable $\bm U(t)$.

At time $t=0$, this random variable has the initial distribution $p(\bm 
F)$. For subsequent times $t_1 < t_2 < ... < t_m$ a  \emph{trajectory} is 
defined as a sequence of temporal realisations $\bm u_1, ..., \bm u_m$  of $\bm 
U(t)$. It represents one possible evolution according to random additions of 
noise in each time step and is visualised in Fig~\ref{fig:trajectory}. 
Importantly, each image $\bm u_i$ in a trajectory only 
depends on $\bm u_{i-1}$. This implies that the corresponding conditional 
transition probabilities fulfil the Markov property
\begin{align}
    \label{eq:markov_property}
    p(\bm u_i | \bm u_{i-1}, ..., \bm u_0) \;=\; p(\bm u_i | \bm u_{i-1}) \, . 
\end{align}
Here, we consider the probability of observing $\bm u_i$ as a realisation of 
$\bm U(t)$ at time $t_i$ given $\bm U(t_{i-1})=\bm u_{i-1}$.
Thus, the stochastic forward evolution is a Markov process \cite{Ga85} and we 
can write the probability density of the trajectory in terms of the transition 
probabilities from~\eqref{eq:markov_property} and the initial distribution 
$p(\bm u_0)=p(\bm F)$: 
\begin{align}
    \label{eq:density_fwdtrajectory}
    p(\bm u_0, ..., \bm u_m) \;=\; p(\bm u_0) \prod_{i=1}^m  p(\bm u_i | \bm 
    u_{i-1}) 
    \, .
\end{align}
This property is also integral to establishing central architectural properties 
of our generalised scale-space in Section~\ref{sec:probdiffscale}. In contrast 
to our earlier conference publication~\cite{Pe23}, we consider a more general 
transition probability than the original model of \citet{SMDH13}. Relying on 
the model of~\citet{KSBH21}, we use Gaussian distributions of the type
\begin{align}
    \label{eq:forwardtransition}
    p(\bm u_i | \bm  u_{i-1}) \;=\; \mathcal{N}\left(\alpha_i \, \bm 
    u_{i-1}, \,
    \beta_i^2 \bm I \right)\, . 
\end{align}
Since $\bm I \in \R^{n \times n}$ denotes the unit matrix, the covariance 
matrix of this multivariate Gaussian is diagonal. Thus, for every pixel $j$, 
we consider independent, identically distributed 
Gaussian noise with 
mean $\alpha_i u_{i-1,j}$ and standard deviation $\beta_i$. Overall, the 
forward process has the free parameters $\alpha_i > 0$ and $\beta_i \in (0,1)$. 
In 
practice, these parameters can be learned or chosen by a user. Often $\alpha_i$ 
is also defined as a function of $\beta_i$, which we discuss in more detail in 
Section~\ref{sec:probdiffscale}. 

\subsection{The Reverse Process}

\citet{SMDH13} motivate the reverse process by a partial differential 
equation (PDE) that is associated to the forward process. In particular, 
they rely on the results of \citet{Fe49}. These require the existence of  the 
stochastic moments
\begin{align}
    m_k(\bm u_t, t) \;=\; \lim_{h \rightarrow 0} \frac{1}{h} \int p(\bm u_{t+h} 
    , 
    \bm u_t) (\bm u_{t+h} - \bm u_t)^k \intd \bm u_{t+h}
\end{align}
with $k\in \{1,2\}$. Under this assumption, the probability density of 
the Markov process from Eq.~\eqref{eq:density_fwdtrajectory} is a solution of 
the partial differential equation
\begin{align}
    \label{eq:fellerfwd}
    \frac{\partial}{\partial t} p \;=\; \frac{1}{2} 
    \frac{\partial^2}{\partial \bm u_t \partial \bm u_t} \left(m_2(\bm u_t, t) 
    p(\bm u_\tau, \bm u_t) \right) 
    + 
    \frac{\partial}{\partial \bm 
        u_t} \left( m_1(\bm u_t, t) p(\bm u_\tau, \bm u_t) \right) \, .
\end{align}
Here, $p(\bm u_\tau, \bm u_t)$ denotes the probability density for a transition 
from $\bm u_\tau$ to $\bm u_t$ with $\tau < t$.
In Section~\ref{sec:osmosis}, we use the fact that 
Eq.~\eqref{eq:fellerfwd} is a drift-diffusion equation to discuss connections 
to osmosis filtering.

For practical purposes, it is important that Feller has proven that a solution 
of Eq.~\eqref{eq:fellerfwd} also solves the backward equation
\begin{align}
    \label{eq:fellerback}
    \frac{\partial}{\partial \tau} p \;=\; \frac{1}{2} m_2(\bm u_\tau, \tau) 
    \frac{\partial^2}{\partial \bm u_\tau \partial \bm u_\tau } p + 
    m_1(\bm u_\tau, \tau) \frac{\partial}{\partial \bm 
        u_\tau}  p \, .
\end{align}
Here, the backward perspective is obtained due to the exchange of roles of the 
earlier time $\tau$ with the later time $t$. \citet{SMDH13} exploit the close 
similarity of the backward equation to the forward equation. 
It implies that the reverse process from the normal distribution to the target 
distribution also has Gaussian transition probabilities. However, the mean and 
standard deviation are unknown and are estimated with a neural network instead. 
In particular, the training minimises the cross entropy to the target 
distribution $p(\bm F)$. We discuss the reverse process in more detail in 
Section~\ref{sec:ddpmbackward}.

The capabilities of diffusion probabilistic models go beyond merely using the 
reverse process to sample from the target distribution. 
Additionally, it is possible to condition this distribution with side 
information such as partial image information or textual descriptions of the 
image content. This is useful for restoring missing image parts with 
inpainting~\cite{RBLE+22,SMDH13} or for text-to-image models~\cite{RBLE+22}.
However, our main focus are theoretical properties of multiple different 
forward processes. Since the estimation of the parameters for the reverse 
process is not relevant for our contributions, we refer to  
\cite{HJA20,KSBH21,RBLE+22,SDME21} for more details. 

%%%%%%%%%%%%%%%%%%%%%%%%%%%%%%%%%%%%%%%%%%%%%%%%%%%%%%%%%%%%%%%%%%%%%%%%%%%%%%%%

\section{Generalised Diffusion Probabilistic Scale-Spaces}
\label{sec:probdiffscale}

In our previous conference publication~\cite{Pe23}, we introduced scale-space 
properties for the original forward diffusion probabilistic model 
(DPM) of Sohl-Dickstein et al.~\cite{SMDH13}. We generalise these results in 
Section~\ref{sec:ddpmforward} to a wider variety of noise schedules. Moreover, 
we introduce a generalised scale-space theory for the corresponding backward 
direction in Section~\ref{sec:ddpmbackward}. Finally, we address the recent 
inverse heat dissipation~\cite{RHS23} and blurring diffusion models 
\cite{HS22} in Section~\ref{sec:blurringdiffusion}.

Before we discuss scale-space properties, we need to establish some 
preliminaries that allow us to rewrite transition probabilities in a 
useful way. The transition probabilities from 
Eq.~\eqref{eq:forwardtransition} allow us to express the random variable at 
time $t_i$ in terms of the random variable at time $t_{i-1}$ according to
\begin{align}
    \label{eq:ddpmforwardrandomvar}
    \bm U_i  \; = \; \alpha_i \, \bm U_{i-1} + \beta_i  \, \bm G \, .
\end{align}
Here, $\bm G$ denotes Gaussian noise from the standard normal distribution 
$\mathcal{N}(\bm 0, \bm I)$. This generalises the model of \citet{SMDH13} who 
use $\alpha_i = \sqrt{1-\beta_i^2}$. \citet{KSBH21} also investigate 
variance-exploding diffusion \cite{SE19} with $\alpha_i^2 = 1$. We discuss both 
types of models in Section~\ref{sec:ddpmforward}.

\medskip
\begin{proposition}[Transition Probability from the Initial Distribution] 
    We can directly transition from $\bm U_0$ to $\bm U_i$ by
    \begin{align}
        \label{eq:ddpmonestep}
        \bm U_i \; = \; \bm U_0 \prod_{\ell = 1}^{i} \alpha_\ell  + 
        \left(\sum_{k=1}^{i-1} \beta_k \prod_{\ell = k+1}^{i} 
        \alpha_\ell + \beta_i \right) \bm G
    \end{align}
\end{proposition}

\begin{proof}
    For $i=1$, the statement is fulfilled according to
    \begin{align}
        \bm U_{1} \;=\; \alpha_{1} \, \bm U_{0} + \beta_1  \, \bm G \, .
    \end{align}
    We prove the statement by induction. Applying the hypothesis for the step 
    from 
    $i$ to $i+1$, we obtain
    \begin{align}
        \bm U_{i+1} &\;=\; \alpha_{i+1} \, \bm U_{i} + \beta_{i+1}  \, \bm G \\
        &\;=\;  \alpha_{i+1} \Bigg( \bm U_0 \prod_{\ell = 1}^{i} \alpha_\ell  
        + \Bigg(\sum_{k=1}^{i-1} \beta_k \prod_{\ell = k+1}^{i} 
        \alpha_\ell + \beta_i \Bigg) \bm G \Bigg)+ \beta_{i+1} \bm G \\
        &\;=\; \bm U_0 \prod_{\ell = 1}^{i+1} \alpha_\ell  + 
        \left(\sum_{k=1}^{i} \beta_k \prod_{\ell = k+1}^{i+1} 
        \alpha_\ell + \beta_{i+1} \right) \bm G \, .
    \end{align}
\end{proof}
\noindent
Thereby, we have established the transition probability from time $0$ to time 
$t_i$ as
\begin{align}
    \label{eq:ddpmonestepdensity}
    p(\bm u_i | \bm u_0) = \mathcal{N}(\lambda_i \bm u_0, \gamma_i^2 \bm I) 
\end{align}
where the mean $\lambda_i$ and standard deviation $\gamma_i$ of this 
multivariate Gaussian distribution are
\begin{align}
    \label{eq:gamma}
    \lambda_i :=  \prod_{\ell = 1}^{i} \alpha_\ell \, , \qquad
    \gamma_i := \sum_{k=1}^{i-1} \beta_k \prod_{\ell = k+1}^{i} 
    \alpha_\ell + \beta_i \, .
\end{align}
An interesting special case of the proposition above arises for the 
parameter choice $\alpha_i = \sqrt{1-\beta_i^2}$ of the variance-preserving 
case \cite{SMDH13}. \citet{HJA20} have shown that under this condition, the 
transition probability becomes
\begin{align}
    \label{eq:starttot}
    p(\bm u_i | \bm  u_0) \;=\; \mathcal{N} \left(\sqrt{\prod_{j=1}^i (1 - 
        \beta_j^2)} \, 
    \bm u_0, \, \bm I - \prod_{j=1}^i (1-\beta_j^2) \bm I \right)\, . 
\end{align}
These insights are helpful for establishing a generalised scale-space 
theory for diffusion probabilistic models.

\begin{figure*}[t]
    \centering
    \includegraphics[width=0.8\linewidth]{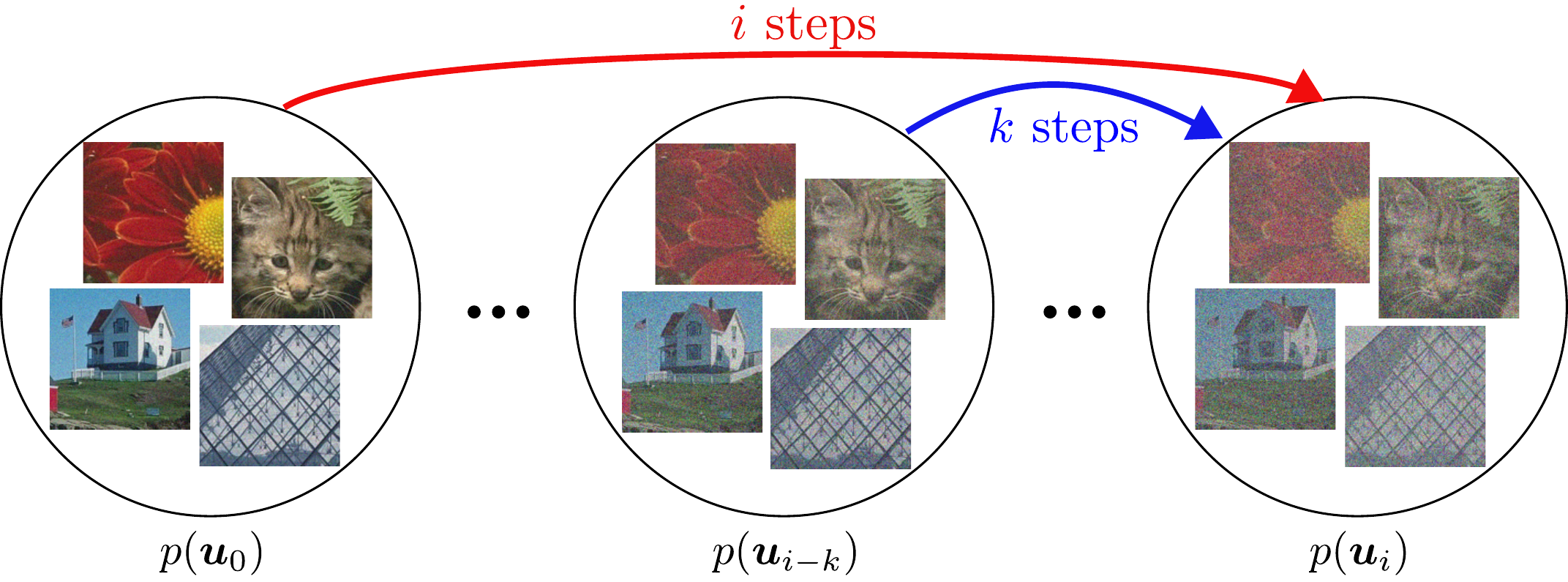}
    \caption{\textbf{Semigroup Property for Forward DPM.} Due to the Markov 
        property, each intermediate scale $i$ can be reached either from the 
        training distribution $p(\bm u_0)$ in $i$ steps, or from $p(\bm 
        u_{i-k})$ 
        in $k$ steps. Note that this property does not apply to individual 
        images 
        as in classical scale-spaces. Instead, it refers to probability 
        distributions  which are visualised by samples from four different 
        trajectories.
        \label{fig:semigroup}}
\end{figure*}

\subsection{Generalised-Scale-Space Properties for Forward DPM}
\label{sec:ddpmforward}

In the following, we propose central architectural properties for a generalised 
DPM scale-space and also discuss invariances. To this end, we consider the
sequence of the marginal distributions of the random variable $\bm U(t)$.
These can be obtained by integrating over all possible paths from the starting
distribution to scale $i$ according to 
\begin{align}
    p(\bm u_i) \;=\; \int p(\bm u_0, ..., \bm u_i) \,  
    \intd \bm u_0 \cdots \intd \bm u_{i-1} \, . 
\end{align}
Thus at each scale $i$, we consider the marginal distribution of $\bm U(t_i)$. 
Individual images are samples from the distributions at a given scale. 

\subsubsection*{Property 1: Initial State}
By definition, the initial distribution for the Markov process is the 
distribution $p(\bm F)$ of the training database. Thus, it also defines the 
initial state $p(\bm u_0)$ of the scale-space.

%........................................................................

\subsubsection*{Property 2: Semigroup Property} 
One central architectural property of scale-spaces is the ability to 
recursively construct a coarse scale from finer scales, i.e. the path from the 
initial state can be split into intermediate scales. This concept has been 
already established by \citet{Ii62} in the pioneering works on Gaussian 
scale-spaces.
Intuitively, diffusion probabilistic models fulfil this property 
since they are Markov processes. The property is visualised in 
Fig.~\ref{fig:semigroup}.

\medskip
\begin{proposition}[Semigroup Property] 
    The distribution $p(\bm u_i)$ at scale $i$ can be reached equivalently in 
    $i$ 
    steps from $p(\bm u_0)$ or in $\ell$ steps from $p(\bm u_{i-\ell})$. 
\end{proposition}

\begin{proof}
    The probability density of the forward trajectory is defined in a recursive 
    way in Eq.~\eqref{eq:density_fwdtrajectory}. Thus, we have to show that 
    this property also carries over to the marginal distributions of the 
    scale-space.
    We can reach $p(\bm u_{i})$ either directly from $\bm u_0$ or from an 
    intermediate scale $\ell$ by using the definition of the joint probability 
    density of the Markov process:
    \begin{align}
        p(\bm u_i) &\;=\;  \int p(\bm u_0) \prod_{j=1}^{i} p(\bm 
        u_j |\bm  u_{j-1}) \,  \intd \bm u_0 \cdots \intd \bm 
        u_{i-1} \, \\
        &\;=\;  \int p(\bm u_{i-\ell}) \prod_{j=i-\ell+1}^{i} p(\bm 
        u_j |\bm  u_{j-1}) \,  \intd \bm u_{i-\ell} \cdots \intd \bm 
        u_{i-1} \, . 
    \end{align}
\end{proof}

%........................................................................

\subsubsection*{Property 3: Lyapunov Sequences} 
In classical scale-spaces (e.g. with diffusion),  Lyapunov sequences quantify 
the change in the evolving image with increasing scale parameter. They 
constitute a  measure of image simplification~\cite{We97} in terms of monotonic 
functions. In practice, they often represent the information content of an 
image at a given scale. Here, we define a Lyapunov sequence on the evolving 
probability density instead.

To this end, we consider the conditional entropy of the random variable $\bm 
U_i$ at time $t_i$ given the random variable $\bm U_0$.
It constitutes a measure for the gradual removal 
of the image information from the initial distribution $p(\bm u_0)$.
\medskip

\begin{proposition}[Increasing Conditional Entropy]  The conditional entropy
    \label{prop:conditionalentropy}
    \begin{align}
        \label{eq:condentropy}
        H_p(\bm U_i | \bm U_0) &\;=\; - \iint p(\bm u_i, \bm u_0) \ln p(\bm 
        u_i | \bm u_0) \, \intd \bm u_0 \intd \bm u_i  \, .
    \end{align}
    increases with $i$ under the assumption $\beta_j \in 
    (0,1)$ for all $j$ with $\beta_{j+1} \geq (1-\alpha_{j+1}) \gamma_j$ with 
    $\gamma_j$ as defined in Eq.~\eqref{eq:gamma}. 
\end{proposition}
\begin{proof}
    We can reduce the problem of showing that the conditional entropy is 
    monotonically increasing to a statement on the differential entropy of 
    $p(\bm 
    u_i | \bm u_0)$ since
    \begin{align}
        H_p(\bm U_i | \bm U_0)
        &\;=\; \int p(\bm u_0)  \Big(\underbrace{- \int p(\bm u_i | \bm u_0) 
            \ln p(\bm u_i | \bm u_0) \intd \bm u_i}_{=: H_p(\bm W_i)} \Big) 
        \intd \bm  u_0  \geq H_p(\bm U_{i-1} | \bm U_0) \\
        & \Leftrightarrow  H_p(\bm W_i) \geq H_p(\bm W_{i-1}) \, .
    \end{align}
    According to Eq.~\eqref{eq:ddpmonestepdensity}, $\bm W_i$ 
    is from $\mathcal{N}(\lambda_i \bm u_0, \gamma_i^2 \bm I)$. Therefore, the 
    entropy of 
    $\bm W_i$ only depends on the covariance matrix $\gamma_i \bm I$ and yields
    \begin{align}
        H_p(\bm W_i) \;=\; \frac{1}{2} \ln\Big(\big(2\pi \conste \, 
        \gamma_i^2\big)^n\Big) \, .
    \end{align}
    Thus, the entropy is increasing if $\gamma_{i+1} \geq \gamma_i$. 
    Furthermore, 
    due to Eq.~\eqref{eq:gamma} we have 
    \begin{align}
        \label{eq:gammarecursive}
        \resizebox{.925\hsize}{!}{$\displaystyle{
                \gamma_{i+1} = \sum_{k=1}^{i} \beta_k \prod_{\ell 
                    = k+1}^{i+1} \hspace{-2mm}
                \alpha_\ell + \beta_{i+1} 
                = \alpha_{i+1} \left(\sum_{k=1}^{i-1} \beta_k \prod_{\ell = 
                k+1}^{i} 
                \hspace{-2mm}  \alpha_\ell + \beta_i\right) + \beta_{i+1} = 
                \alpha_{i+1} 
                \gamma_i + 
                \beta_{i+1} \,
                .}$}
    \end{align}
    Since $\beta_i > 0$ and $\gamma_i > 0$, we require 
    \begin{align}
        \alpha_{i+1} \gamma_i + \beta_{i+1} \geq \gamma_i \quad \Leftrightarrow 
        \quad 
        \beta_{i+1} 
        \geq (1-\alpha_{i+1}) \gamma_i \, .
    \end{align}
    Again, we can also consider  $\alpha_i = \sqrt{1-\sigma_i}$, $\beta_i = 
    \sqrt{\sigma_i}$ as in \cite{SMDH13}. This gives us more concrete 
    expressions 
    for $\gamma_i$ according to Eq.~\eqref{eq:starttot}. With this we obtain
    \begin{align}
        \gamma_{i+1} = \sqrt{1 - \prod_{j=1}^{i+1} (1-\beta_j^2)} = \sqrt{1 - 
            (1-\beta_{i+1}^2) 
            \prod_{j=1}^{i}  (1-\beta_j^2)} > \gamma_i \, .
    \end{align}
    Since $\beta_{i+1} \in (0,1)$, this holds without further conditions with 
    the 
    noise schedule of Sohl-Dickstein et al.~\cite{SMDH13}.
\end{proof}

%........................................................................

\subsubsection*{Property 4: Permutation Invariance} 
The 1-D drift diffusion process acts independently on each image pixel. 
Therefore, the spatial configuration of the pixels does not matter for the 
process. In the following we provide formal arguments for a permutation 
invariance of all distributions created by the trajectories of the 
drift-diffusion. 

Let $P(\bm f)$ denote a permutation function that arbitrarily shuffles the 
position of the pixels in the image $\bm f$ from the initial database. In 
particular, such permutations also include cyclic translations as well as 
rotations by $90^\circ$ increments.

\medskip

\begin{proposition}[Permutation Invariant Trajectories] 
    Let $\bm u_0$ denote an image from the initial distribution and $\bm v_0 := 
    P(\bm u_0)$ its permutation. Then, any trajectory $\bm v_0, ... \bm v_m$
    obtained from the process in Eq.~\eqref{eq:ddpmforwardrandomvar} is given by
    $\bm v_i = P(\bm u_i)$ for a trajectory $\bm u_0, ..., \bm u_m$ starting 
    with the original image $\bm u_0$.
\end{proposition}
\begin{proof}
    Consider the transition of $\bm v_{i-1}$ to $\bm v_i$ via a realisation 
    $\bm g_i$ of the random variable $\bm G$ in 
    Eq.~\eqref{eq:ddpmforwardrandomvar}.
    Then $\tilde{\bm g}_i := P^{-1}(\bm g_i)$ is also from the distribution 
    $\mathcal{N}(\bm 0, \bm I)$ and we can obtain  $\bm u_i$ from $\bm u_{i-1}$
    by using this permuted transition noise.
    Since $\bm v_0 = P(\bm u_0)$ holds by definition, we can inductively show 
    the claim by considering
    \begin{equation}
        \bm v_i \;=\; \alpha_i \bm v_{i-1} + \beta_i \bm g_i  
        \;=\; 
        \alpha_i P(\bm u_{i-1}) + \beta_i P(\tilde{\bm g}_i) 
        \;=\; 
        P(\bm u_i) \, .
    \end{equation}
\end{proof}

Any permutation $P$ is a bijection. Thus every trajectory from a permuted image
corresponds exactly to one trajectory starting from the original image. This 
directly implies permutation invariance of the corresponding distributions.

\medskip

\begin{corollary}[Permutation Invariant Distributions] 
    Let $\bm u_0$ denote an image from the initial distribution and $\bm v_0 := 
    P(\bm u_0)$ its permutation. Then, for any image $\bm v$ from $p(\bm v_i)$
    there exists exactly one image $\bm u$ from $p(\bm u_i)$ such that $\bm v = 
    P(\bm u)$.
\end{corollary}

%........................................................................
\subsubsection*{Property 5: Steady State}
The steady state distribution for $i \rightarrow \infty$ is a multivariate 
Gaussian distribution $\mathcal{N}(\bm 0, \beta^2 \bm I)$ with mean $\bm 0$ and 
a covariance matrix~$\beta^2 \bm I$ with $\beta < 2$. Convergence  to a 
noise 
distribution is a cornerstone of diffusion probabilistic models and thus 
well-known. For the sake of completeness we provide formal arguments for this 
property. Moreover, we verify that for the noise schedule of Sohl-Dickstein 
et al.~\cite{SMDH13}, we obtain $\beta=1$. Thus, we verify that their steady 
state is the normal distribution. 

\medskip
\begin{proposition}[Convergences to a Normal Distribution] 
    \label{prop:ddpmsteady}
    Let $\alpha_i$ and $\beta_i$ be bounded from above by 
    $a, b \in (0,1)$, i.e. $\alpha_i \in (0,a]$ and $\beta_i \in (0,b]$. 
    Moreover, let the assumptions of 
    Proposition~\ref{prop:conditionalentropy} be 
    fulfilled. 
    Then, for $i \rightarrow \infty$, the forward process from 
    Eq.~\eqref{eq:density_fwdtrajectory} converges to a normal distribution 
    $\mathcal{N}(\bm 0, \gamma^2 \bm I)$ with $\gamma \leq 
    \frac{b}{1-a}$.
\end{proposition}
\begin{proof}
    According to Eq.~\eqref{eq:ddpmonestep} and Eq.~\eqref{eq:gamma}, we have
    \begin{align}
        \bm U_{i} \;=\; \lambda_i \bm U_0 + \gamma_i \bm G
    \end{align}
    with $\bm G$ from $\mathcal{N}(\bm 0, \bm I)$. For $i \rightarrow \infty$, 
    we can immediately conclude $\lambda_i \rightarrow 0$ for the 
    mean of the steady state distribution since it is the product of $i$ 
    factors $\alpha_\ell < 1$. 
    
    Under the assumptions of  Proposition~\ref{prop:conditionalentropy}, we 
    have already shown that $\gamma_i$ is increasing. Now let us consider 
    the boundedness of $\gamma_i$, starting with definition 
    Eq.~\eqref{eq:gamma} and using the assumptions $0<\alpha_i \leq a$ and 
    $0<\beta_i \leq b$:
    \begin{align}
        \ \gamma_i &= \sum_{k=1}^{i-1} \beta_k \prod_{\ell = k+1}^{i} 
        \alpha_\ell + \beta_i 
        \; \leq \; \sum_{k=1}^{i-1}  b \prod_{\ell = k+1}^{i} 
        a + b \\
        &= b \left(\sum_{k=1}^{i-1} \prod_{\ell = k+1}^{i}a + 1 \right)  
        \; = \; b \sum_{k=0}^{i-1}  a^k \;  \xrightarrow{i \to \infty} \;  
        \frac{b}{1-a} \, .
    \end{align}  
    Overall, this shows that for $i \to \infty$, every trajectory converges 
    to 
    $\gamma \bm G$ with a $\gamma \leq \frac{b}{1-a}$ and
    $\bm G$ from $\mathcal{N}(\bm 0, \bm I)$.
\end{proof}

We have only specified an upper bound for the variance $\gamma^2$ of 
the steady state distribution so far. The original DPM model \cite{SMDH13} 
with $\alpha_i = \sqrt{1-\beta_i^2}$ does not only act as an example that
verifies reasonable parameter choices are possible under the assumptions of 
our proposition. Additionally, we can also explicitly  infer that 
$\gamma=1$ in this special case.
Due to 
\begin{align}
    \gamma_i \;=\; \sqrt{1-\prod_{j=1}^{i} (1-\beta_j^2)}
\end{align}
and $0 < 1-\beta_j^2 < 1$, we obtain $\gamma_i \xrightarrow{i \to \infty} 
1$. 
Thus, the original DPM model convergences to the standard normal distribution.

Note that the variance-exploding model is not covered by the steady state 
criterion without additional assumptions. Due to the parameter choice 
$\alpha_\ell = 1$, the sequence 
$\gamma_i$ is given by
\begin{align}
    \gamma_i = \sum_{k=1}^{i} \beta_k \, .
\end{align}
As the name of the model suggests, the variance is thus not necessarily 
bounded 
for $i \to \infty$. For special choices of $\beta_i$, e.g. $\beta_i = 
\beta_0^{i-1}$ with $\beta_0 = 1$, we can however still get convergence to 
a 
normal distribution with a fixed variance. In the case of the example it 
would 
be $(1-\beta_0)^{-1}$.
Also note that due to $\alpha_\ell=1$ in Eq.~\eqref{eq:gamma}, 
$\lambda_i=1$ for all $i$. Thus, the mean remains constant in this case.

The noisy steady state marks a clear 
difference to traditional scale-spaces. For instance, diffusion scale-spaces on 
images~\cite{We97} converge to a flat steady state instead. However, the 
new class of diffusion probabilistic scale-spaces still underlies the 
same core concept: 
It removes information from the initial state recursively and 
hierarchically, leading to a state of minimal information w.r.t. the initial 
distribution.

\subsection{Generalised Scale-Space Properties for Reverse DPM}
\label{sec:ddpmbackward}

\citet{SMDH13} argue via results of \citet{Fe49} that for infinitesimal 
$\beta_t$, the distribution of forward and reverse trajectories becomes 
identical. However, these results are also tied in an inverse proportional way 
to the length of the trajectory. For arbitrary small $\beta_t$, the number 
of steps goes to infinity. For such a case of identical distributions, our 
results for the forward process would carry over to the reverse process: 
Initial and steady state are swapped and our Lyapunov sequences are 
decreasing instead of increasing. The remainder of the properties carry over 
verbatim.

In practice, however, the time-discrete reverse process used for DPMs is 
an approximation. Neural networks estimate the 
parameters for this reverse process. In the following, we comment on properties 
that can be established under these conditions.

To this end, we consider the reverse process of \citet{SMDH13} and denote its 
distributions with $q$. It takes the 
normal distribution $\mathcal{N}(\bm 0, \bm I)$ as a starting distribution 
$q(\bm u_M)$. Transitions in the reverse 
direction from $t_i$ to $t_{i-1}$ fulfil
\begin{align}
    \label{eq:reversetransition}
    q(\bm  u_{i-1} | \bm u_i) \;=\; \mathcal{N}\left(\bm \mu(\bm u_i, i), \bm 
    \Sigma (\bm u_i, i) \right) \, .
\end{align}
While these learned distributions are still Gaussian, they are significantly 
more complex than in the forward process. Both the learned mean and variance do 
not reduce to common scalars for all pixels. Furthermore, they depend on the 
current time step and on the image $\bm u_i$ itself. Therefore, we can 
establish less properties for the reverse process than before, but central 
ones still carry over.

%........................................................................

\subsubsection*{Property 1: Normal Distribution as Initial State}
By definition, the distribution $q(\bm u_M)$ at time $t_M$ is given by 
$\mathcal{N}(\bm 0, \bm I)$.

%........................................................................
\subsubsection*{Property 2: Semigroup Property}
The distribution $q(\bm u_i)$ at time $t_i$, $i<M$ can be reached in $M-i$ 
steps from $p(\bm u_M)$ or in $M-\ell-i$ steps from $p(\bm u_{M-\ell})$ 
with $M - \ell > i$. 
Since the Markov property is fulfilled, the proof for the semigroup property is 
analogous to the forward case in Section~\ref{sec:ddpmforward}. 

%........................................................................
\subsubsection*{Property 3: Lyapunov Sequence}
As a byproduct of their derivation of conditional bounds for the reverse 
process, \citet{SMDH13} have already concluded that both the entropy $H_q(\bm 
u_i)$ and the conditional entropy $H_q(\bm u_0 | \bm u_i)$ are decreasing for 
the backward direction, i.e. 
\[
H_q(\bm u_{i-1}) \leq H_q(\bm u_{i}) \, .
\]
This is plausible, since the evolution starts with noise, a state of maximum 
entropy, and sequentially introduces more structure to it.

%........................................................................
\subsubsection*{Property 4: Steady State} 
DPM reverse processes have the goal to enable sampling from the 
unknown initial distribution $p(\bm u_0)$. Convergence to this distribution is 
only guaranteed for the ideal case with identical distributions $p$ and $q$. 
However, even if this is not strictly fulfilled, the parameters $\bm 
\mu$ and $\bm \Sigma$ are chosen such that they maximise a lower bound of the 
log likelihood 
\begin{align}
    \int q(\bm u_0) \log p(\bm u_0) \, \intd \bm u_0 \, . 
\end{align}
In this sense, the reverse process approximates the distribution 
of the training data at time $t=0$.

\medskip

Property 4 from Section~\ref{sec:ddpmforward}, the permutation invariance, does 
in general not apply to the reverse process. Since the parameters $\bm \mu$ 
and $\bm \Sigma$ depend on the previous steps $\bm u_i$ of the trajectory, the 
configuration of the pixels matters. Invariances will only be present if the 
network that estimates the parameters enforces them in its architecture.

\subsection{Generalised Scale-Space Properties for Blurring Diffusion}
\label{sec:blurringdiffusion}

Inverse heat dissipation~\cite{RHS23} and blurring diffusion~\cite{HS22} models
do not solely rely on adding noise in order to destroy features of the original 
image. Instead, they combine it with a deterministic homogeneous diffusion 
filter~\cite{Ii62} to gradually blur this image. Such diffusion filters are 
well-known as the origin of scale-space theory~\cite{Ii62} and also constitute 
a special case of the osmosis filters we consider in more detail in 
Section~\ref{sec:osmosis}. First, we discuss equivalent formulations of 
blurring diffusion in the spatial and transform domain~\cite{RHS23,HS22}. These 
allow us to transfer our scale-space results for DPM 
to this new setting.

A discrete linear diffusion operator can be interpreted either as the 
discretisation of a continuous time evolution described by a partial 
differential equation (see also Section~\ref{sec:osmosis}) or as a Gaussian 
convolution. In the following we consider only greyscale images with $N=n_x 
n_y$ pixels. Colour images can be processed by filtering each channel 
separately. \citet{RHS23} use an operator ${\bm A_i = \exp(t_i \bm \Delta)}$, 
where $\bm \Delta \in \R^{N \times N}$ is a discretisation of the Laplacian 
$\Delta u = \partial_{xx} u + \partial_{yy} u$ with reflecting boundary 
conditions.  By adding Gaussian noise, they turn the deterministic diffusion 
evolution into a probabilistic process given by
\begin{align}
    \label{eq:ihdforward}
    \ut{i} \;=\; \bm A_i \ut{0} + \bm \eps_i, \qquad \bm \eps_i \in 
    \mathcal{N}(\bm 0, \beta_i^2  \bm I)
\end{align}
In particular, they make use of a change of basis. To this end, let 
$\bm V \in \R^{N \times N}$ denote the basis transform operator of the 
orthogonal discrete cosine transform (DCT). Furthermore, we use the notation 
$\tilde{\bm u} = \bm V \bm u$ to denote the DCT representation of a spatial 
variable $\bm u$. Then, the diffusion operator $\bm A_t = \bm V^\top \bm B_t 
\bm V$ reduces to a diagonal matrix $\bm B_t = \textnormal{diag} (\bm 
\alpha_1, ..., \bm \alpha_N)$ in the DCT domain. The 
entries of $\bm B_t$ result from the eigendecomposition of the Laplacian. 
Let the vector index $j \in \{1,...,N\}$ correspond to the position 
$(k,\ell) \in \{0,...,n_x-1\} \times \{0,...,n_y-1\}$ in the two-dimensional 
frequency domain. Then the entries of $\bm B_t$ are given by
\begin{align}
    \alpha_{j} \;=\; \exp\left(- t \pi^2  \left( \frac{k^2}{n_x} + 
    \frac{\ell^2}{n_y}  
    \right)\right)  \, .
    \label{eq:laplace_dct}
\end{align}
Note that for the frequency $(k,\ell)^\top=(0,0)^\top$, $\alpha_0 = 1$. Hence, 
the average grey value of the image is preserved. Moreover, for all 
$j \neq 0$, we have $\alpha_{j} < 1$ for $t > 0$. 
Additionally, since $\bm V$ is orthogonal and 
$\bm V^T \bm V = \bm I$, it also preserves Gaussian noise: For a sample $\bm 
\eps$ from the normal distribution $\mathcal{N}(\bm 0, \bm I)$, the transform 
$\tilde{\bm \eps}$ is from the same distribution. 
These properties are important for some of our 
scale-space considerations later on.

\citet{RHS23} proposed the DCT representation of their process for a fast 
implementation. However, \citet{HS22} used 
it to derive a more general formulation in the DCT domain. Since $\bm B_t$ is 
diagonal, a step from time $t_i$ to time $t_{i+1}$ can be considered for 
individual scalar frequencies $j$:
\begin{align}
    \tilde{u}_{i+1,j} \;=\; \alpha_{i,j} \tilde{u}_{i,j} + \beta_i \eps \, .
\end{align}
Here, $\eps$ is from $\mathcal{N}(0,1)$.
In contrast to \citet{RHS23}, they do not limit the noise variance $\sigma_t$
to minimal observation noise, but also allow to choose a noise schedule as in 
the DPM from Section~\ref{sec:ddpmforward}. This formulation is 
particularly useful for us since it constitutes another 1-D drift diffusion 
process as in Eq.~\eqref{eq:ddpmforwardrandomvar}. It allows us to transfer 
some of our previous findings to the new setting. However, note that compared 
to the model in Eq.~\eqref{eq:ddpmforwardrandomvar} we operate on frequencies 
instead of images as realisations of the random variable. Moreover, each 
frequency has its own individual parameters $\alpha_{i,j}$. In addition, 
$\beta_i$ 
could be made frequency specific. However, in practice, \citet{HS22} choose the 
same $\beta_i$ for all frequencies.
This also entails that the Markov criterion for this process in DCT space is 
given by 
\begin{align}
    \label{eq:markovdct}
    p(\tilde{\bm u}_{i+1} |\tilde{\bm 
        u}_{i}) \;=\; \mathcal{N}(\tilde{\bm u}_i | 
    \bm B_i \tilde{\bm u}_{i-1}, \beta_i^2 \bm I) \, .
\end{align}
In the following we consider a scale-space that is defined by the 
marginal distributions of the random variable $\bm U_t := \bm V^\top \tilde{\bm 
    U}_t$, i.e. the backtransform of the trajectories in DCT space. Note that 
    every 
trajectory in the frequency domain has exactly one corresponding trajectory in 
the spatial domain. Therefore, we can argue equivalently in 
the DCT domain or the spatial domain, depending on what is more convenient. As 
in Section~\ref{sec:ddpmforward}, the process starts with the distribution of 
the training data or, respectively, the distribution of its discrete cosine 
transform.

%........................................................................

\subsubsection*{Property 1: Initial State}
By definition, the distribution $p(\bm u_0)$ at time $t_0=0$ is given by the 
distribution $p(\bm F)$ of the training data.

%........................................................................

\subsubsection*{Property 2: Semigroup Property} 
The distribution $p(\bm u_i)$ at scale $i$ can be reached equivalently in $i$ 
steps from $p(\bm u_0)$ or in $\ell$ steps from $p(\bm u_{i-\ell})$. 
The Markov property is fulfilled in the DCT domain. Thus,  the proof for the 
semigroup property is analogous to the DPM model in 
Section~\ref{sec:ddpmforward}. 
For each intermediate scale we can switch back to the spatial domain by 
multiplication with $\bm V^\top$.

%........................................................................

\subsubsection*{Property 3: Lyapunov Sequence}
To establish this information reduction property, we require an analogous 
statement to Eq.~\eqref{eq:ddpmonestep}. By using a 
similar induction proof for each frequency, we obtain
\begin{align}
    \label{eq:bdonestep}
    \tilde{\bm U}_i \; = \; \bm M_i  \tilde{\bm U}_0 + 
    \bm \Sigma_i \bm G \, .
\end{align}
Here, $\bm G$ is from $\mathcal{N}(\bm 0, \bm I)$, $\bm M_i = 
\textnormal{diag}(\lambda_{i,0},...,\lambda_{i,N})$ with 
\begin{align}
    \label{eq:lambdadct}
    \lambda_{i,j} \;=\; \prod_{\ell = 1}^{i} \alpha_{\ell, j} \, , 
\end{align}
and  $\bm \Sigma_u = \textnormal{diag}(\gamma_{i,0},...,\gamma_{i,N})$
with
\begin{align}
    \label{eq:gammadct}
    \gamma_{i,j} \;=\;  \left(\sum_{k=1}^{i-1} \beta_{i} \prod_{\ell = 
        k+1}^{i} 
    \alpha_{\ell,j} + \beta_{i} \right) \, . 
\end{align} 
This is a frequency dependent analogue statement to the direct transition from 
time $t_0$ to time $t_i$ in the DPM setting in Eq.~\eqref{eq:gamma}.

\medskip
\noindent
\begin{proposition}[Increasing Conditional Entropy]  The conditional entropy
    \label{prop:bdlyapunov}
    \begin{align}
        \label{eq:condentropy_bd}
        H_p(\tilde{\bm U}_i | \tilde{\bm U}_0) &\;=\; - \int \int p(\tilde{\bm 
            u}_i, \tilde{\bm u}_0) \ln p(\bm 
        \tilde{\bm u}_i | \bm \tilde{\bm u}_0) \, \intd \bm \tilde{\bm u}_0 
        \intd \bm \tilde{\bm u}_i 
    \end{align}
    increases with $i$ under the assumption that for all frequencies $j$ and 
    $\beta_{i} \in (0,1)$ we have $\beta_{i+1} \geq (1-\alpha_{i+1,j}) 
    \gamma_{i,j}$. 
\end{proposition}
\begin{proof}
    As in Section~\ref{sec:ddpmforward}, the statement is equivalent to showing 
    that the entropy $H_p(\bm W_i)$ of the distribution $p(\tilde{\bm u}_i | 
    \tilde{\bm u}_0)$
    is increasing. Thus, we need to show that $H_p(\bm W_i) \geq H_p(\bm 
    W_{i-1})$.
    The probability distribution of $p(\tilde{\bm u}_i | \tilde{\bm u}_0)$ can 
    be inferred from Eq.~\eqref{eq:ddpmonestepdensity}, but it is more complex 
    than in the DPM case due to the frequency dependent parameters. 
    
    Fortunately, the entropy of the multivariate Gaussian distribution  
    $\mathcal{N}(\bm M_t \tilde{\bm u}_0, \bm \Sigma_t)$ only depends on its 
    covariance matrix $\bm \Sigma_t$ and is given by
    \begin{align}
        H_p(\bm W_i) = \frac{1}{2} \ln\Big(\big(2\pi \conste \, \big)^n 
        \det{\bm \Sigma_i}\Big) = \frac{1}{2} \ln\Bigg(\big(2\pi \conste 
        \,\big)^n \prod_{j=1}^N \gamma_{i,j} \Bigg)
    \end{align}
    If $\gamma_{i+1,j} \geq \gamma_{i,j}$ holds for all frequencies $j$, the 
    entropy is increasing. For a fixed frequency $j$, we can transfer the 
    previous result from Eq.~\eqref{eq:gammarecursive} to the scalar setting:
    \begin{align}
        \gamma_{i+1, j} = \alpha_{i+1,j} \gamma_{i,j} + \beta_{i+1} \, .
    \end{align}
    Since $\beta_i > 0$ and $\gamma_i > 0$, we require 
    \begin{align}
        \alpha_{i+1,j} \gamma_{i,j} + \beta_{i+1} \geq \gamma_{i,j} \quad 
        \Leftrightarrow \quad 
        \beta_{i+1} 
        \geq (1-\alpha_{i+1,j}) \gamma_{i,j}
    \end{align}
    This is a direct extension of our previous result in 
    Section~\ref{sec:ddpmforward} to the frequency setting.
\end{proof}

Note that there are also previous results for deterministic diffusion filters
that use entropy as a Lyapunov sequence~\cite{We97}. However, there the entropy 
is defined on the pixel values of the image instead of an evolving probability 
distribution. Individual trajectories rely on a deterministic 
diffusion filter. However, due to the added noise, the entropy statements of 
classical scale-spaces do not transfer to the trajectories of blurring 
diffusion.

%........................................................................

\subsubsection*{Property 3: Preservation of the Average Grey Value}
The DPM scale-space from Section~\ref{sec:ddpmforward} ensures convergence to
Gaussian noise with zero mean, independently of the initial image. This 
requires the mean of the image to change. The behaviour of blurring diffusion 
scale-spaces is significantly different.

\medskip
\begin{proposition}[Preservation of the Average Grey Value] 
    \label{prop:bdavg}
    Let $\bm u_0$ denote an image from the initial distribution $p(\bm f)$.
    Then, all images in the trajectory $\bm u_0, ..., \bm u_m$ have the same 
    average grey value.
\end{proposition}

This statement directly follows from an observation on the spatial version of 
the process. In every step, we add Gaussian noise with mean zero to a diffusion 
filtered image. Since diffusion filtering preserves the average grey value 
\cite{We97}, this also holds for blurring diffusion. In colour images, the 
preservation of the average colour value applies for each channel.

%........................................................................

\subsubsection*{Property 4: Rotation Invariance}
In contrast to the forward DPM scale-space, blurring diffusion takes into 
account the neighbourhood configuration in the spatial domain due to the 
blurring of the homogeneous diffusion operator. Thus, permutation invariance 
does not apply to blurring diffusion scale-spaces.

However, in a space-continuous setting, the diffusion operator is rotationally 
invariant. The same applies to Gaussian noise samples: Under rotation, they 
remain samples from the same noise distribution. In the fully discrete setting, 
this rotation invariance is typically partially lost since only $90^\circ$ 
rotations align perfectly with the pixel grid. However, this depends on the 
concrete implementation of the process. From this observation, we can directly 
deduce the following statement.

\medskip
\begin{proposition}[Rotation Invariant Trajectories] 
    Let $\bm u_0$ denote an image from the initial distribution and $\bm v_0 := 
    R(\bm u_0)$ a rotation by a multiple of $90^\circ$. Then, any trajectory 
    $\bm v_0, ... \bm v_m$
    obtained from the process in Eq.~\eqref{eq:ddpmforwardrandomvar} is given by
    $\bm v_i = R(\bm u_i)$ for a trajectory $\bm u_0, ..., \bm u_m$ starting 
    with the original image $\bm u_0$.
\end{proposition}

%........................................................................

\subsubsection*{Property 5: Steady State}
Due to the preservation of the average grey value in Property 3, a blurring 
diffusion scale-space cannot converge to a noise distribution with zero 
mean unless the initial image already had a zero mean. Images typically 
have nonnegative pixel values (e.g. a range of $[0,255]$ or $[0,1]$).
Thus, a zero mean would be only possible for a flat image or after a 
transformation to a range that is symmetric to zero (e.g. $[-1,1]$). We do not 
make any such assumption. 
However, for the sake of simplicity, we consider greyscale 
images in the following. For colour images, the same statements apply for each 
channel. 

\medskip
\begin{proposition}[Convergence to a Mixture of Normal Distributions] 
    Under the assumptions of Proposition~\ref{prop:bdlyapunov},
    $\alpha_{i,j} \in (0,a_j]$ with $0<a_j<1$, $\beta_i \in (0,b]$ with 
        $0<b<1$, and $i \rightarrow \infty$, the 
    forward  process from Eq.~\eqref{eq:density_fwdtrajectory} converges to a 
    mixture of 
    normal distributions $\mathcal{N}(\mu_k \bm I, \bm \Sigma)$ with 
    $\bm \Sigma = \textnormal{diag}(\sigma_1,...,\sigma_N)$ and $\sigma_i 
    \leq \frac{b}{1-a_j}$. The value $\mu_k$ with $k \in \{0,...,n_f\}$ 
    assumes all possible average grey levels from the training database. 
\end{proposition}

\begin{proof}
    Consider a trajectory starting from an arbitrary image $\bm u_0$ from the 
    training database with mean $\mu$.
    As in the proof for Proposition~\ref{prop:bdlyapunov} the findings from 
    Eq.~\eqref{eq:ddpmonestep} and Eq.~\eqref{eq:gamma} for DPM carry over to 
    our setting in the scalar case for each frequency $j$:
    \begin{align}
        \tilde{u}_{i,j} \;=\; \lambda_{i,j} \tilde{u}_{0,j} + \gamma_{i,j}  \eps
    \end{align}
    with $\eps$ from $\mathcal{N}(0, 1)$. For the convergence of 
    $\lambda_{i,j}$ we have to consider the product of all frequency specific 
    $\alpha_{i,j}$ according to Eq.~\eqref{eq:lambdadct}.
    Here we have the special case $\alpha_{i,0} = 1$ for the lowest frequency, 
    i.e. $\lambda_{i,0}=1$ for all $i$. This is consistent with our previous 
    findings in Proposition~\ref{prop:bdavg}: The lowest frequency represents 
    the average grey level, which remains unchanged.
    
    For $j>0$, we have  $\alpha_{i,j} < 1$ and thus $\lambda_{i,j} \rightarrow 
    0$ for $i \rightarrow \infty$. Thus, for all other frequencies the 
    contribution of $\bm u_i$ vanishes and only its mean $\mu$ is preserved.
    The convergence of  $\gamma_{i,j}$ to $\sigma_i \leq \frac{b}{1-a_j}$ 
    is analogous to the proof of Proposition~\ref{prop:ddpmsteady}. This 
    determines the standard deviation of the noise for each frequency.
\end{proof}

\medskip
Overall, blurring diffusion constitutes a scale-space that resembles the DPM 
scale-space in its architectural properties. The key difference lies in the 
incorporation of 2-D neighbourhood relationships between pixel values in the 
spatial domain. In the DCT domain, this translates to an individual set of 
process parameters for each frequency.

%%%%%%%%%%%%%%%%%%%%%%%%%%%%%%%%%%%%%%%%%%%%%%%%%%%%%%%%%%%%%%%%%%%%%%%%%%%%%%%%%%%%%%%

\section{Diffusion Probabilistic Models and Osmosis}
\label{sec:discussion}

All diffusion probabilistic models considered in 
Section~\ref{sec:probdiffscale} have in common that they are connected to 
drift-diffusion processes. 
In image processing with  partial differential equations (PDEs), osmosis 
filters proposed by \citet{WHBV13} are successfully applied to 
image editing and restoration 
tasks~\cite{DMM18,VHWS13,WHBV13,PCCSW19}. 
Since they share their origin with diffusion probabilistic models, namely 
the Fokker-Planck equation \cite{Ri84}, we discuss connections in the 
following. First, we briefly review the PDE formulation of 
osmosis filtering. Afterwards, we discuss common properties as well as 
differences between both classes of models. Finally, we compare all four models 
experimentally.

\subsection{Continuous Osmosis Filtering}
\label{sec:osmosis}

Unlike the first part of the manuscript, we consider a grey value image as a 
function $f : \Omega \rightarrow \mathbb{R}_+$ defined on the image domain  
$\Omega \subset \mathbb{R}^2$. It maps each coordinate $\bm x \in \Omega$ to a 
positive grey value. In the following we limit our description to grey value 
images for the sake of simplicity. The filter can be extended to colour or 
arbitrary other multi-channel images (e.g. hyperspectral) by applying it 
channel-wise.

As in diffusion  probabilistic models, osmosis filters consider the 
evolution of an image $u: \Omega \times [0,\infty) \rightarrow \mathbb{R}_+$ 
over time.
However, here the evolution is entirely deterministic. Its initial state is 
given by a starting image $f$, i.e. for all $\bm x$ we have $u(\bm x, 0) = 
f(\bm x)$. In addition, the second major factor that influences the 
evolution is the \emph{drift vector field} $\bm d: \Omega \rightarrow 
\mathbb{R}^2$ that can be chosen independently of the initial image and is 
typically used for filter design.

Given these two degrees of freedom, the image evolution fulfils the 
PDE~\cite{WHBV13}
\begin{align}
    \partial_t u &\;=\; \Delta u - \Div(\bm d u) & \quad \text{on } 
    \Omega 
    \times (0,T] \, .  \label{eq:ospde}
\end{align}
At the boundaries of $\Omega$, reflecting boundary conditions avoid any exchange
of information with areas outside of the image domain. There is a direct 
connection of this model to the inverse heat dissipation~\cite{RHS23} and 
blurring diffusion~\cite{HS22}: All of these processes build on 
homogeneous diffusion~\cite{Ii62}, which is a special case of osmosis with $\bm 
d = \bm 0$.

However, a non-trivial drift vector field enables to describe evolutions that
do not merely smooth an image. The Laplacian $\Delta u$, which 
corresponds to the diffusion part of the PDE in Eq.~\eqref{eq:ospde}, 
represents a symmetric exchange of information between neighbouring pixels. 
This symmetry can be broken by the drift component $-\Div(\bm d u) = 
-\partial_x (d_1 u) - \partial_y (d_2 u)$. This is also vital for our own
goal of relating osmosis image evolutions to diffusion probabilistic 
models and will allow us to introduce stochastic elements without any need to 
modify the PDE above.

\subsection{Relating Osmosis and Diffusion Probabilistic Models}
\label{sec:common}

As noted in the previous section, for $\bm d=\bm 0$, osmosis is directly 
connected
to blurring diffusion. For the original DPM model, the connection is less 
obvious, but equally close. To this end, it is instructive to consider a 1-D 
osmosis process. Its evolution is described by the PDE
\begin{align}
    \partial_t u &\;=\;  \partial_{xx} u - \partial_x (d u)  & \quad 
    \text{on } 
    \Omega \subset{\R}
    \times (0,T] \, .  \label{eq:1dospde}
\end{align} 
A comparison with the PDE formulation of the DPM forward process in 
Eq.~\eqref{eq:fellerfwd} reveals a significant structural similarity. Both 
equations describe an evolution w.r.t. the time $t$: 1-D osmosis considers an 
evolving image $u$, while the DPM evolution is defined on a probability 
density $p$. The derivative $\partial_{\bm u_t}$ w.r.t. positions of individual 
particles in the diffusion probabilistic model corresponds to the spatial 
derivative $\partial_x$ in Eq.~\eqref{eq:1dospde}. The moments correspond to 
diffusivity and drift factors for the osmosis PDE and thus we can interpret 1-D 
osmosis as the deterministic counterpart to the probabilistic drift-diffusion 
process of DPM.

While both models are derived from the same physical principle, a 1-D osmosis 
filter would typically act on a 1-D signal instead of an image. In this 1-D 
setting, it also considers neighbourhood relations and exchanges information 
between neighbouring signal elements. In contrast, the 1-D drift-diffusion of 
DPM acts on individual pixels of a 2-D image which correspond to positions of 
individual particles. Thus, 1-D osmosis creates evolving 1-D signals, while a 
DPM trajectory consists of 2-D images.

For a meaningful comparison between the properties of osmosis filtering and the 
three diffusion probabilistic models, we require 2-D osmosis as in 
Eq.~\eqref{eq:ospde}. Moreover, we need to add a stochastic component. 
Ideally, every trajectory of this new probabilistic osmosis process should 
inherit the theoretical properties of the deterministic model. Inverse heat 
dissipation~\cite{RHS23} can be seen as a similar approach in that it provides 
a stochastic counter part to homogeneous diffusion~\cite{Ii62}. However, there 
the addition of noise marks a clear departure from the PDE model and also 
implies that properties of the deterministic model do not carry over to the 
trajectory.

Instead of adding noise, we design the drift vector field $\bm d$ in such a 
way that the osmosis evolution naturally converges to a noise sample. 
\citet{WHBV13} have shown that osmosis preserves the average grey value. Thus, 
for a given image $f$ from the training data, we choose a noise image $ 
\eps$, where at every location $x \in \Omega$, $\eps(x)$ is from  
$\mathcal{N}(\mu_{f}, \sigma^2)$. Here, $\mu_{f}$ is the average grey value of 
$f$. Note that osmosis requires positive images.  In practice, we clip the 
noise sample to $[10^{-6},1]$. While this truncates the tails of the Gaussian 
distribution, it has little impact for small standard deviation~$\sigma$.

Now we can use another result of \citet{WHBV13} and consider the 
\emph{compatible} case for osmosis. Our noise sample $\eps$ acts as 
guidance 
image for the osmosis process. The corresponding canonical drift vector field 
is defined by
\begin{equation}
    \bm d = \frac{\grad \eps}{\eps} \label{eq:canonical}
\end{equation}
with the spatial gradient $\grad$. Since $\mu_f=\mu_\eps$, the 
osmosis 
process converges to  
\begin{equation}
    \frac{\mu_{f}}{\mu_{\eps}}  \eps = \eps \, .
\end{equation}
Considering intermediate results of the evolution at times $t_i$ yields a 
trajectory $u_0, ..., u_m$ of the probabilistic osmosis process with $u_0 = f$.
Even though the intermediate marginal probabilities or the transition 
probabilities are not known, the process starts with the target distribution 
$p(\bm F)$ and converges to an approximation to a normal distribution. In 
the following, we can empirically compare trajectories of this osmosis process 
to trajectories of diffusion probabilistic models.

%%%%%%%%%%%%%%%%%%%%%%%%%%%%%%%%%%%%%%%%%%%%%%%%%%%%%%%%%%%%%%%%%%%%%%%%%%%%%%%%%%%%%%%

\begin{figure*}[!th]
    \begin{center}
        \tabcolsep2pt
        \newlength{\imgwidth}
        \setlength{\imgwidth}{0.155\textwidth}
        \begin{tabular}{cccccc}
            $t=1$ & $t=5$ & $t=10$ & $t=100$ & $t=1000$ & $t=10000$ 
            \\[1.5mm]
            \multicolumn{6}{c}{(a) diffusion probabilistic model 
                (\citet{SMDH13})}\\
            \includegraphics[width=\imgwidth]{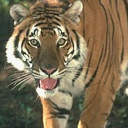}
            &
            \includegraphics[width=\imgwidth]{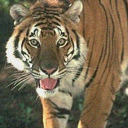}
            &
            \includegraphics[width=\imgwidth]{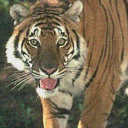}
            &
            \includegraphics[width=\imgwidth]{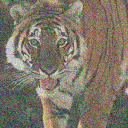}
            &
            \includegraphics[width=\imgwidth]{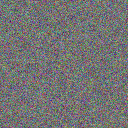}
            &
            \includegraphics[width=\imgwidth]{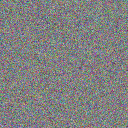}
            \\[1.5mm]
            \multicolumn{6}{c}{(b) inverse heat dissipation (\citet{RHS23})}\\
            \includegraphics[width=\imgwidth]{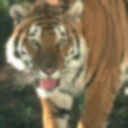}
            &
            \includegraphics[width=\imgwidth]{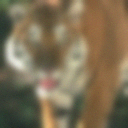}
            &
            \includegraphics[width=\imgwidth]{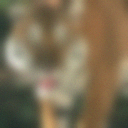}
            &
            \includegraphics[width=\imgwidth]{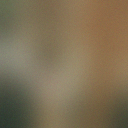}
            &
            \includegraphics[width=\imgwidth]{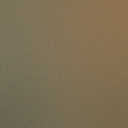}
            &
            \includegraphics[width=\imgwidth]{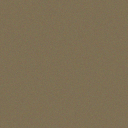}
            \\[1.5mm]
            \multicolumn{6}{c}{(c) homogeneous diffusion (\citet{Ii62})}\\
            \includegraphics[width=\imgwidth]{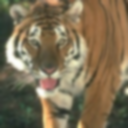}
            &
            \includegraphics[width=\imgwidth]{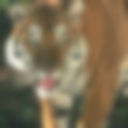}
            &
            \includegraphics[width=\imgwidth]{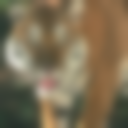}
            &
            \includegraphics[width=\imgwidth]{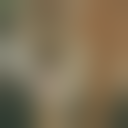}
            &
            \includegraphics[width=\imgwidth]{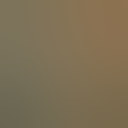}
            &
            \includegraphics[width=\imgwidth]{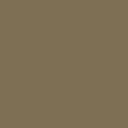}
            \\[1.5mm]        
            \multicolumn{6}{c}{(d) blurring diffusion (\citet{HS22})}\\
            \includegraphics[width=\imgwidth]{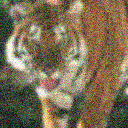}
            &
            \includegraphics[width=\imgwidth]{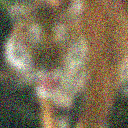}
            &
            \includegraphics[width=\imgwidth]{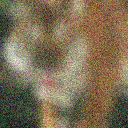}
            &
            \includegraphics[width=\imgwidth]{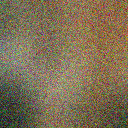}
            &
            \includegraphics[width=\imgwidth]{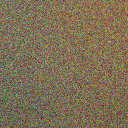}
            &
            \includegraphics[width=\imgwidth]{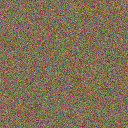}
            \\[1.5mm]
            \multicolumn{6}{c}{(e) osmosis (\citet{WHBV13})}\\[0.5mm]
            % \hline   
            \includegraphics[width=\imgwidth]{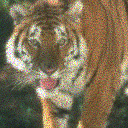}
            &
            \includegraphics[width=\imgwidth]{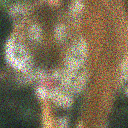}
            &
            \includegraphics[width=\imgwidth]{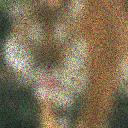}
            &
            \includegraphics[width=\imgwidth]{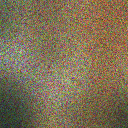}
            &
            \includegraphics[width=\imgwidth]{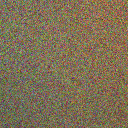}
            &
            \includegraphics[width=\imgwidth]{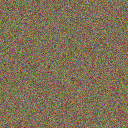}
            \\            
        \end{tabular}
    \end{center}%
    
    \caption{\label{fig:scalespaces} \textbf{Visual comparison of 
            trajectories.} The diffusion probabilistic model (a) behaves 
        distinctively different compared to the other approaches since it does 
        not 
        perform blurring in the image domain. Due to the minimal amounts of 
        added 
        noise, inverse heat dissipation (b) closely resembles homogeneous 
        diffusion 
        (c). With a suitable noise schedule, blurring diffusion (d) closely 
        resembles osmosis filtering (e).  
    }
\end{figure*}

\subsection{Comparing Osmosis Filters and Diffusion Probabilistic Models}

For concrete experiments, we require a discrete implementation of the 
continuous osmosis model from Section~\ref{sec:osmosis}. As \citet{VHWS13}, we 
use a stabilised BiCGSTAB solver~\cite{Me15}. For the noise guidance images, we 
use a standard deviation of $\sigma=0.1$. All experiments are conducted on the 
Berkeley segmentation dataset \emph{BSDS500}~\cite{AMFM11}. 

We compare to the three models from Section~\ref{sec:probdiffscale}.
According to Eq.~\eqref{eq:ddpmforwardrandomvar}, we implement forward 
DPM~\cite{SMDH13} by successively adding Gaussian noise with the standard 
parameter choice ${\alpha_i = \sqrt{1-\beta_i^2}}$ and $\beta_i=0.1$. Inverse 
heat dissipation~\cite{RHS23} and blurring diffusion~\cite{HS22} can be 
implemented in many equivalent ways. We based our implementation on the 
reference code of~\citet{RHS23}, which implements the Laplacian in the DCT 
domain according to Eq.~\eqref{eq:laplace_dct}. For inverse heat dissipation, 
we use the standard parameter $\beta_i = 0.01$. For blurring diffusion, we 
choose $\beta_i = 0.1$ such that it coincides with the guidance noise of our 
probabilistic osmosis. Furthermore, we also include homogeneous 
diffusion~\cite{Ii62} without any added noise in our comparisons.

\subsubsection{Visual Comparison}

Fig.~\ref{fig:scalespaces} reveals visual similarities and differences between 
model trajectories. All five models successively remove features of the initial 
training sample: They are drowned out by noise, by 
blur, or a combination of both. For the probabilistic models, we have shown 
that this information reduction is quantified by entropy-based Lyapunov 
sequences. A similar statement holds for images in an osmosis 
trajectory. As shown by \citet{Sc18}, the relative entropy of $u$ w.r.t. the 
noise sample $\eps$ is increasing:  
\begin{align}
    L(t) \;:=\; -\int_\Omega  u(\bm x,t)\ln \left( \frac{u(\bm x,t)}{w(\bm 
        x)}\right) \, \intd\bm x \, .
\end{align}
This also reflects the transition from the initial image $\bm f$ to the steady 
state $\eps$. Similar statements on an unconditioned entropy apply to 
homogeneous diffusion \cite{We97}. 

Notably, only DPM does not preserve the average  colour value of the 
initial image and converges to noise with mean zero. For visualisation 
purposes, the images of the DPM trajectory have therefore been 
affinely mapped to $[0,1]$. DPM is also the only process that does not take 
neighbourhood relations between pixels into account. Therefore, edge 
features, such as the stripe pattern in Fig.~\ref{fig:scalespaces}(a), remain 
sharp until they are completely overcome by noise. 

This observation directly results from the 1-D drift-diffusion: DPM models the 
microscopic aspect of Brownian motion with colour values as particle positions. 
All other models consider the macroscopic aspect of drift-diffusion instead 
which considers colour values as particle concentrations per pixel cell. 
Consequentially, all other models apply 2-D blur in the image domain. Due to 
the very small amount of observation noise added by \citet{RHS23}, the 
trajectory of heat dissipation is visually very similar to homogeneous 
diffusion. 

Similarly, osmosis and blurring diffusion lead to visually almost identical 
trajectories. They mainly differ in the way how noise is added: Blurring 
diffusion uses explicit addition, while osmosis transitions to noise due to the 
drift vector field. In the following, we verify these observations 
quantitatively.

\begin{figure*}[t]
    \centering
    \tabcolsep2pt
    \begin{tabular}{cc}
        \includegraphics[width=0.49\linewidth]{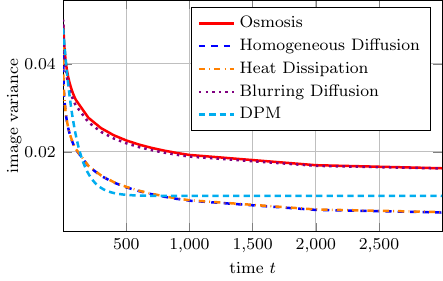}
        &
        \includegraphics[width=0.49\linewidth]{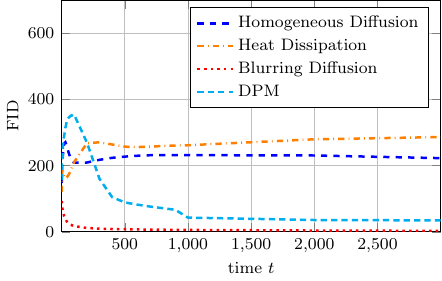}\\
        (a) variance & (b) FID
        \\[1mm]
    \end{tabular}
    \caption{\textbf{Quantitative Comparison of Diffusion Probabilistic 
                Models and Model-based Filters}.
        Both the variance evolution over time in (a) and the FID w.r.t. the 
        osmosis 
        distributions in (b) suggest that DPM differs significantly from the 
        classical diffusion and osmosis filters. Heat dissipation approximates 
        diffusion, while blurring diffusion approximates osmosis.
        \label{fig:variance}}
\end{figure*}

\subsubsection{Variance Comparison}

On the entire BSDS500 database, we evaluate the evolution of the image variance 
over time in Fig.~\ref{fig:variance}(a). As expected, the pairs homogeneous 
diffusion/heat dissipation and osmosis/blurring diffusion exhibit very similar 
evolutions of the variance. On the way to the flat steady state, homogeneous 
diffusion and heat dissipation approach zero variance. Osmosis and blurring 
diffusion converge to a noise variance defined by the input parameters while 
DPM very slowly converges to the standard distribution. These observations 
coincide with the expectations from our theoretical results.

\subsubsection{FID Comparison}

Additionally, we can judge the similarity of intermediate distributions
in the scale-space with the Fr\'echet-Inception distance (FID) \cite{HRUNH17}. 
It is widely used to judge the quality of generative models in terms of the 
approximation quality towards the target distribution. We use the 
implementation clean-fid \cite{PZZ22} that avoids discretisation artefacts
due to sampling and quantisation. Note that 
we measure the FID of probabilistic osmosis distributions relative to results 
of the other four models. Thus, a low FID indicates how closely each filter 
approximates osmosis.

Fig.~\ref{fig:variance}(b) also confirms our previous hypothesis: Heat 
dissipation and blurring diffusion consistently differ most from osmosis 
results since they rely mostly on blur and not on noise. DPM comes close to 
osmosis in its noisy steady state, but differs significantly in the initial 
evolution due the lack of 2-D smoothing. Blurring diffusion approximates 
osmosis surprisingly closely over the whole evolution.
\medskip

\citet{HS22} have found that heat dissipation improves the overall quality of 
generative models compared to DPM. Blurring diffusion yields even better 
results. Given our findings, we can interpret these observations from a 
scale-space perspective: The integration of 2-D neighbourhood relationships in 
the scale-space evolution is important for good diffusion probabilistic 
models. 
However, also the addition of sufficient amounts of stochastic perturbations is 
vital. Overall, recent advances can be interpreted as an increasingly 
accurate approximation of osmosis filtering, with an approximation to diffusion 
scale-spaces as an intermediate model. Using such an approximation instead of 
directly applying 2-D osmosis is convenient due to the more straightforward 
relation to the reverse process.

%%%%%%%%%%%%%%%%%%%%%%%%%%%%%%%%%%%%%%%%%%%%%%%%%%%%%%%%%%%%%%%%%%%%%%%%%%%%%%%%%

\section{Conclusions and Outlook}
\label{sec:conclusion}

Inspired by diffusion probabilistic filters, we have proposed the first 
class of generalised stochastic scale-spaces that describe evolutions of 
probability distributions instead of images. 
While the setting differs significantly from 
classical scale-spaces, central properties such as gradual, quantifiable 
simplification and causality still apply. These results suggest that in general,
sequential generative models from deep learning are closely connected to 
scale-space theory. Therefore, we hope that in the future, the scale-space 
community will benefit from the discovery of new scale-spaces that might also 
be used in different contexts. In particular, existing generative models are 
mostly focused on the steady states as the practically relevant output. The 
intermediate results of the associated scale-spaces could however also be 
useful in future applications.

On the flip side, trajectories of recent diffusion probabilistic models 
approximate well-known classical scale-space evolutions. This suggests that in 
the opposite direction, the deep learning community can potentially benefit 
from existing knowledge about scale-spaces by incorporating them into deep 
learning approaches.

\backmatter

\bmhead{Acknowledgments}

I thank  Vassillen Chizhov, Michael Ertel, Kristina Schaefer, Karl Schrader, 
and Joachim Weickert for fruitful discussions and advice.
I gratefully acknowledge the stimulating research environment of the GRK 
2853/1 ``Neuroexplicit Models of Language, Vision, and Action'', funded by the 
Deutsche Forschungsgemeinschaft (DFG, German Research Foundation) under project 
number 471607914.

\bibliography{bib.bib}% common bib file

\end{document}